\documentclass[11pt,reqno,tikz]{amsart}
\usepackage{amsmath,amsthm,amssymb,bm,bbm,dsfont,braket}
\usepackage[mathscr]{eucal}
\usepackage{amsaddr}
\usepackage{upgreek}
\usepackage[left=2cm,right=2cm,top=2cm,bottom=2cm]{geometry}
\usepackage{mathtools}\mathtoolsset{centercolon}
\usepackage{cite}

\usepackage[shortlabels]{enumitem}
\usepackage{setspace}
\usepackage{graphicx}
\usepackage{verbatim}
\usepackage{float}
\usepackage{placeins}
\usepackage{array}
\usepackage{booktabs}
\usepackage{multicol,multirow,tabularx}
\usepackage{threeparttable}
\usepackage[update,prepend]{epstopdf}
\usepackage{amsfonts,amssymb,dsfont,xfrac,nicefrac,comment}
\usepackage[abs]{overpic}

\makeatletter
\def\adl@drawiv#1#2#3{%
	\hskip0
	\tabcolsep
	\xleaders#3{#2 0\@tempdimb #1{1}#2 0.5\@tempdimb}%
	#2\z@ plus1fil minus1fil\relax
	\hskip0\tabcolsep}
\newcommand{\cdashlinelr}[1]{%
	\noalign{\vskip\aboverulesep
		\global\let\@dashdrawstore\adl@draw
		\global\let\adl@draw\adl@drawiv}
	\cdashline{#1}
	\noalign{\global\let\adl@draw\@dashdrawstore
		\vskip\belowrulesep}}
\makeatother

\usepackage[usenames,dvipsnames]{color}
\usepackage{colortbl}
\usepackage{arydshln}
\usepackage[hidelinks,linktocpage=true]{hyperref}
\hypersetup{
	unicode=false,          
	pdftoolbar=true,        
	pdfmenubar=true,        
	pdffitwindow=false,     
	pdfstartview={FitH},    
	pdftitle={My title},    
	pdfauthor={Author},     
	pdfsubject={Subject},   
	pdfcreator={Creator},   
	pdfproducer={Producer}, 
	pdfkeywords={keyword1} {key2} {key3}, 
	pdfnewwindow=true,      
	colorlinks=true,        
	linkcolor=Red,          
	citecolor=ForestGreen,  
	filecolor=Magenta,      
	urlcolor=BlueViolet,    
}
\usepackage{doi}
\usepackage{url}
\usepackage{caption, subcaption}
\usepackage{enumitem}
\usepackage{shadethm}
\usepackage{tikz,pgf}
\usepackage{witharrows}

\usepackage{cleveref}

\makeatletter
\ifx\@NODS\undefined%

\let\mathbb=\mathds
\else%
\fi
\makeatother

\def\d{{\text {\rm d}}}

\DeclareMathOperator{\Tr}{Tr}

\DeclareMathOperator{\spec}{spec}
\def\I{\mathds{1}}

\newcommand{\proj}[1]{\left\{#1\right\}} 

\newcommand{\be}{{\mathbf e}}

\def\0{{\mathbf{0}}}
\def\1{{\mathbf{1}}}
\def\2{{\mathbf{2}}}
\def\3{{\mathbf{3}}}
\def\4{{\mathbf{4}}}
\def\5{{\mathbf{5}}}
\def\6{{\mathbf{6}}}

\def\7{{\mathbf{7}}}
\def\8{{\mathbf{8}}}
\def\9{{\mathbf{9}}}

\def\be{\begin{equation}}
\def\ee{\end{equation}}
\def\bea{\begin{eqnarray}}
\def\eea{\end{eqnarray}}

\theoremstyle{plain}
\newshadetheorem{theo}{Theorem}
\newshadetheorem{prop}{Proposition} 
\newshadetheorem{lemm}[prop]{Lemma} 
\newshadetheorem{theorem}{Theorem}

\theoremstyle{definition}

\theoremstyle{remark}

\makeatletter
\newcommand{\opnorm}{\@ifstar\@opnorms\@opnorm}
\newcommand{\@opnorms}[1]{%
	$\left|\mkern-1.5mu\left|\mkern-1.5mu\left|
	#1
	\right|\mkern-1.5mu\right|\mkern-1.5mu\right|$
}
\newcommand{\@opnorm}[2][]{%
	\mathopen{#1|\mkern-1.5mu#1|\mkern-1.5mu#1|}
	#2
	\mathclose{#1|\mkern-1.5mu#1|\mkern-1.5mu#1|}
}
\makeatother

\usepackage{tikz}
\usetikzlibrary{arrows.meta} 
\tikzset{>={Latex[length=4,width=4]}} 
\usetikzlibrary{calc}

\colorlet{mylightblue}{blue!5!white}
\colorlet{mydarkblue}{blue!30!black}
\colorlet{myblue}{blue!50!black}
\colorlet{myred}{red!50!black}
\colorlet{mydarkred}{red!30!black}
\colorlet{mydarkgreen}{green!30!black}

\newcommand{\sh}{\kern-0.08em$^\textbf{\#}$\hspace{-3pt}}
\renewcommand{\b}{\kern-0.06em$\flat$}

\begin{document}

\let\origmaketitle\maketitle
\def\maketitle{
	\begingroup
	\def\uppercasenonmath##1{} 
	\let\MakeUppercase\relax 
	\origmaketitle
	\endgroup
}

\title{\bfseries \Large{ 
The operator layer cake theorem is equivalent to\\ Frenkel's integral formula
		}}

\author{ \normalsize 
{Hao-Chung Cheng}$^{1\text{--}5}$,
{Gilad Gour}$^{6}$,
{Ludovico Lami}$^{7}$,
and
{Po-Chieh Liu}$^{1,2}$
}
\address{\small  	
$^1$Department of Electrical Engineering and Graduate Institute of Communication Engineering,\\ National Taiwan University, Taipei 106, Taiwan\\
$^2$Department of Mathematics and National Center for Theoretic Science, National Taiwan University, Taipei 106, Taiwan\\
$^3$Center for Quantum Science and Engineering,  National Taiwan University, Taipei 106, Taiwan\\
$^4$Hon Hai (Foxconn) Quantum Computing Center, New Taipei City 236, Taiwan\\
$^5$Physics Division, National Center for Theoretical Sciences, Taipei 10617, Taiwan\\
$^6$Faculty of Mathematics, Technion-Israel Institute of Technology, Haifa 3200003, Israel\\
$^7$Scuola Normale Superiore, Piazza dei Cavalieri 7, 56126 Pisa, Italy
}
%


\begin{abstract}
The operator layer cake theorem provides an integral representation for the directional derivative of the operator logarithm in terms of a family of projections [\href{https://arxiv.org/abs/2507.06232}{arXiv:2507.06232}].
Recently, the related work [\href{https://arxiv.org/abs/2507.07065}{arXiv:2507.07065}] showed that the theorem gives an alternative proof to Frenkel's integral formula for Umegaki's relative entropy
[\href{http://doi.org/10.22331/q-2023-09-07-1102}{\textit{Quantum}, 7:1102 (2023)}].
In this short note, we find a converse implication, demonstrating that the operator layer cake theorem is equivalent to Frenkel's integral formula.
\end{abstract}

\maketitle
\vspace{-2.5em}



\section{Introduction} \label{sec:intro}

We consider a finite-dimensional Hilbert space.
For a positive definite operator $B>0$ and a Hermitian operator $H$, we denote the directional derivative of the natural logarithm at $B$ with direction $H$ by
\begin{align}
    \mathrm{D}\log[B](H)
    \coloneq \lim_{t\to 0} \frac{\log(B+tH) - \log B}{t}.
\end{align}
In Ref.~\cite[Theorem B.1]{preparation}, an \emph{operator layer cake theorem} for $\mathrm{D}\log[B](H)$ has been proved, i.e.,
\begin{align} \label{eq:double-side0}
    \mathrm{D}\log[B](H) 
    = \int_0^\infty \proj{ H > \gamma B } \d \gamma
    - \int_{-\infty}^0 \proj{ H \leq \gamma B } \d \gamma,
\end{align}
where $\proj{ H > \gamma B }$ (resp.~$\proj{ H \leq \gamma B } $) is the projection onto the strictly positive part (resp.~non-positive part)
of $H-\gamma B$.
This integral representation finds uses in showing error exponents for quantum packing-type problems such as quantum channel coding~\cite{preparation} as well as for numerous quantum covering-type problems~\cite{sharp25}.

On the other hand, for any $A\geq 0$ and $B>0$, Umegaki introduced the quantum relative entropy~\cite{Ume62}
\begin{align}
    D(A\Vert B)
    \coloneq \Tr\left[ A \left( \log A - \log B \right) + B - A \right],
\end{align}
for which Frenkel established the following integral trace representation~\cite{Fre23}:
\begin{align} \label{eq:Frenkel00}
    D(A\Vert B)
    = \int_{-\infty}^\infty \frac{\d t}{|t|(t-1)^2} \Tr\left[ \left( (1-t)A + t B \right)_{-}\right],
\end{align}
where $(H)_{\pm} \coloneq \frac12\big( \sqrt{H^2} \pm H\big)$ denotes the positive or negative part of a Hermitian operator $H$.
Later, the formula was rewritten in the following form
\cite{hirche2023quantum, Jen24}:
\begin{align} \label{eq:Frenkel0}
    D(A\Vert B)
    = \int_1^\infty \bigg\{ \frac{1}{\gamma} E_{\gamma}(A\Vert B) + \frac{1}{\gamma^2} E_{\gamma}(B\Vert A) \bigg\}\, \d \gamma,
\end{align}
where the quantum hockey-stick divergence for $A, B \geq 0$ with a parameter $\gamma \geq 0$ is defined by
\begin{align}
    E_{\gamma}(A\Vert B) 
    \coloneq \Tr\left[ \left( A-\gamma B\right)_+\right].
\end{align}

In Ref.~\cite[Proposition 4.2]{LHC25_layer_cake}, it was shown that the operator layer cake theorem~\eqref{eq:double-side0}
implies~\eqref{eq:Frenkel0}, providing an alternative proof to Frenkel's integral formula.
In this note, we will show that 
Frenkel's formula~\eqref{eq:Frenkel0} implies a special case of the operator layer cake theorem with any positive direction, i.e.,
\begin{align} \label{eq:single-side0}
    \mathrm{D}\log[B](A) 
    = \int_0^\infty \proj{ A > \gamma B } \d \gamma, \quad \forall\, A\geq 0.
\end{align}
Moreover, we will show that~\eqref{eq:single-side0} implies the general version in~\eqref{eq:double-side0}.
Hence, the operator layer cake theorem~\eqref{eq:double-side0} is equivalent to Frenkel's formula~\eqref{eq:Frenkel0}.

\section{Result and Proof} \label{sec:result}
\begin{prop}
The following statements are equivalent:
\begin{enumerate}[(i)]
    \item Operator layer cake theorem~\cite[Theorem B.1]{preparation}:
    \begin{align} \label{eq:double-side}
    \mathrm{D}\log[B](H) 
    = \int_0^\infty \proj{ H > \gamma B } \d \gamma
    - \int_{-\infty}^0 \proj{ H \leq \gamma B } \d \gamma, \quad \forall\ H = H^\dagger,\ B>0.
    \end{align}

    \item Operator layer cake theorem with positive direction:
    \begin{align} \label{eq:single-side}
    \mathrm{D}\log[B](A) 
    = \int_0^\infty \proj{ A > \gamma B } \d \gamma, \quad \forall\ A\geq 0,\ B> 0.
    \end{align}

    \item Frenkel's integral formula~\cite{frenkel2022integral}:
    \begin{align} \label{eq:Frenkel}
    D(A\Vert B)
    = \int_1^\infty \bigg\{ \frac{1}{\gamma} E_{\gamma}(A\Vert B) + \frac{1}{\gamma^2} E_{\gamma}(B\Vert A) \bigg\}\, \d \gamma
    , \quad \forall\ A\geq 0,\ B>0.
    \end{align}
\end{enumerate}
\end{prop}
\begin{proof}
The implication ``(i) $\Rightarrow$ (ii)'' clearly holds, since $\proj{H\leq \gamma B } = 0$ for any $H\geq 0$ and $\gamma < 0$.
The implication ``(i) $\Rightarrow$ (iii)'' was proved in~\cite[Proposition~4.2]{LHC25_layer_cake} via the fundamental theorem of calculus.
Below, we will show ``(iii) $\Rightarrow$ (ii)'' and ``(ii) $\Rightarrow$ (i)'', completing the equivalence of the three statements.
In the end, we will also provide a proof of the implication ``(iii) $\Rightarrow$ (i)'', although that would not be strictly needed.

Proof of ``(iii) $\Rightarrow$ (ii)'':
For any Hermitian $X$,
\begin{align} \label{eq:derivative_1}
\left.\frac{\d}{\d t} D(A\Vert B + t X )\right|_{t=0}
&= -\Tr\left[ A \cdot \mathrm{D}\log [B](X)\right] + \Tr X.
\end{align}
On the other hand,
\begin{equation}
\begin{aligned}
    &\frac{\d}{\d t} D(A \,\Vert B+tX )\bigg\vert_{t=0} \\
    &\qquad = \frac{\d}{\d t} \int_1^\infty \bigg\{ \frac{1}{ \gamma}E_\gamma(A\Vert B+tX)+
    \frac{1}{\gamma^2}E_\gamma(B+tX\Vert A) \bigg\}\, \d\gamma\,\bigg\vert_{t=0}
    \\
    &\qquad =\frac{\d}{\d t} \int_1^\infty \bigg\{ \frac{1}{ \gamma}\Tr\left[(A-\gamma(B+tX))_+\right] +
    \frac{1}{\gamma^2}\Tr\left[(B+tX-\gamma A)_+\right]\bigg\}\, \d\gamma\,\bigg\vert_{t=0}
    \\
    &\qquad =\frac{\d}{\d t}\left( \int_1^\infty   \frac{1}{ \gamma}\Tr\left[(A-\gamma B+t(-\gamma X))_+\right] \d\gamma +
    \int_0^1 \Tr\left[(B-\gamma^{-1}A+tX)_+\right]\d\gamma\right)\bigg\vert_{t=0}
    \\
    &\qquad \overset{(\dag)}{=} \int_1^\infty   \frac{1}{ \gamma}\cdot\frac{\d}{\d t}\Tr\left[(A-\gamma B+t(-\gamma X))_+\right] \bigg\vert_{t=0}\!\!\d\gamma +
    \int_0^1 \frac{\d}{\d t}\Tr\left[(B-\gamma^{-1}A+tX)_+\right]\bigg\vert_{t=0}\!\!\d\gamma
    \\
    &\qquad = \int_1^\infty   \frac{1}{\gamma}\Tr[-\gamma X\cdot\{A-\gamma B>0\}]\,\d\gamma +
    \int_0^1 \Tr[X\cdot \{B-\gamma^{-1}A \geq 0\}]\,\d\gamma
    \\
    &\qquad = -\int_1^\infty \Tr[X\{A>\gamma B\}]\,\d\gamma +
    \int_0^1 \Tr[X \{\gamma B\geq A\}]\,\d\gamma
    \\
    &\qquad = -\int_0^\infty \Tr[X\{A>\gamma B\}]\,\d\gamma +
    \Tr [X].
    \label{eq:step_projection}
\end{aligned}
\end{equation}
Here, in~($\dag$), we took the derivative inside the integral, applying the dominated convergence theorem. To see why this is possible, we first notice that if we choose $t_0>0$ such that $B+tX>B/2$ for all $|t|<t_0$, then we have $A<\gamma(B+tX)$ for $\gamma>2r$, where $r\coloneqq \|B^{-1/2}AB^{-1/2}\|$. Therefore, the first integral on the left-hand side of ($\dag$) can be rewritten as
\[
\int_1^{2r}    \frac{1}{ \gamma}\Tr\left[(A-\gamma B+t(-\gamma X))_+\right] \d\gamma.
\]
First, the integrand is differentiable at $t=0$ almost everywhere in $\gamma$, except for $\gamma\in \spec(AB^{-1})$, by Lemma~\ref{lemm:differentiability}.
Recall that the function $Y \mapsto \Tr[Y_+]$ is Lipschitz continuous with respect to the trace norm $\|\cdot\|_1$ (see Lemma~\ref{lemm:Lipschitz} below and also \cite[Lemma 2]{Lami_2025}). 
Consequently, for $0<|t|<t_0$, the magnitude of the difference quotient for the first integrand is bounded by
\begin{align}
\frac{1}{|t|\gamma} \left| \Tr[(A-\gamma(B+tX))_+] - \Tr[(A-\gamma B)_+] \right| 
\le \frac{1}{|t|\gamma} \| -\gamma t X \|_1 = \|X\|_1.
\end{align}
Since the integration domain $[1, 2r]$ is compact and the bounding function $\|X\|_1$ is integrable, the Lebesgue Dominated Convergence Theorem justifies interchanging the derivative at $t=0$ and the integral.

Similarly, for the second integral over $[0,1]$, the integrand is differentiable at $t=0$ almost everywhere in $\gamma$ by Lemma~\ref{lemm:differentiability} and the difference quotient is again bounded by $\|X\|_1$, which permits the application of the Lebesgue Dominated Convergence Theorem to the second term as well.


Note that the map $\mathrm{D} \log [B] (\cdot)$ is self adjoint with respect to the Hilbert--Schmidt inner product~\cite{Lie73}.
Therefore, from~\eqref{eq:derivative_1}, we have
\begin{align}
    \Tr\left[ A \cdot \mathrm{D}\log [B](X)\right]
    =
    \Tr [X \cdot \mathrm{D}\log [B](A)]
    =\int_0^\infty \Tr[X\{A>\gamma B\}]\,\d\gamma.
\end{align}

Since the equality holds for any Hermitian $X$, we can conclude that
\begin{align}
    \mathrm{D}\log[B](A) 
    = \int_0^\infty \proj{ A > \gamma B } \d \gamma,
\end{align}
showing the implication ``(iii) $\Rightarrow$ (ii)''.

Proof of ``(ii) $\Rightarrow$ (i)'':
For any $B> 0$ and Hermitian $H$, let $r>\left\|B^{-1/2} H B^{-1/2}\right\|_{\infty}$, where $\|\cdot\|_\infty$ denotes the operator norm. Then, $H+rB>0$. 
We calculate
\begin{equation}
\begin{aligned}
    \mathrm{D}\log [B](H)
    &= \mathrm{D}\log [B](H+rB)-\mathrm{D}\log [B](rB)
    \\
    &= \mathrm{D}\log [B](H+rB)-r\I
    \\
    &\overset{{\eqref{eq:single-side}}}{=}\int_0^\infty \{H+rB>\gamma B\} \,\d \gamma-r\I
    \\
    &=\int_r^\infty \{H+rB>\gamma B\} \,\d \gamma
    + \int_0^r \{H+rB>\gamma B\} \,\d \gamma
    -r\I
    \\
    &=\int_r^\infty \{H+rB>\gamma B\} \,\d \gamma-\int_0^r\{H+rB \leq \gamma B\} \,\d \gamma
    \\
    &=\int_r^\infty \{H+rB>\gamma B\} \,\d \gamma-\int_{-\infty}^r\{H+rB \leq \gamma B\} \,\d \gamma
    \\
    &= \int_0^\infty \{H+rB> (\gamma+r) B\} \,\d \gamma-\int_{-\infty}^0 \{H+rB\leq (\gamma+r)B\} \,\d \gamma
    \\
    &=\int_0^\infty \{H> \gamma B\} \,\d \gamma-\int_{-\infty}^0 \{H\leq \gamma B\} \,\d \gamma.
\end{aligned}
\end{equation}

Proof of ``(iii) $\Rightarrow$ (i)'': 
Let $X$ be a Hermitian matrix, and let $s_0,t_0>0$ be small enough so that $B\pm s_0X,B\pm t_0H>0$.
Then, we have for $s\in[-s_0,s_0]$ and $t\in[-t_0,t_0]$,
\begin{equation}
\begin{aligned}
    \Tr\!\left[ X \cdot \mathrm{D}\log[B](H) \right]
    &= \Tr\!\left[ X \cdot \frac{\d}{\d t}\log(B+tH)\bigg\vert_{t=0} \right]
    \\
    &= \frac{\partial^2}{\partial s\,\partial t}
       \Tr\!\left[(B+sX)\log(B+tH)\right]\bigg\vert_{s=t=0}
    \\
    &= -\frac{\partial^2}{\partial s\,\partial t}
       D(B+sX\Vert B+tH)\bigg\vert_{s=t=0}.
    \label{eq:LC_step1}
\end{aligned}
\end{equation}
Here, as the function $\Tr\left[(B+sX)\log(B+tH)\right]$ is smooth jointly in $(s,t)$, the partial derivatives with respect to $s$ and $t$ are interchangeable, due to Schwarz's theorem.
Next, the directional derivative given in Lemma~\ref{lemm:differentiability} below shows that
\begin{align}
    \frac{\partial}{\partial s}
    E_\gamma(B+sX\Vert B+tH)
    &= \Tr\!\left[
        X\,\{B+sX >\gamma (B+t H)\}
      \right], 
    \label{eq:dE1}\\
    \frac{\partial}{\partial s}
    E_\gamma(B+tH\Vert B+sX)
    &= -\gamma\,\Tr\!\left[
        X\,\{\gamma (B+sX) <B+ tH\}
      \right],
    \label{eq:dE2}
\end{align}
almost everywhere in $\gamma$ for each fixed $s,t$.  Using Frenkel's integral representation~\eqref{eq:Frenkel}, we obtain
\begin{equation}
\begin{aligned}
    \frac{\partial}{\partial s}D(B+sX\Vert B+tH)\bigg\vert_{s=0}
    &= \int_1^\infty \frac{1}{\gamma}
       \frac{\partial}{\partial s}
       E_\gamma(B+sX\Vert B+tH)\bigg\vert_{s=0}\, \d\gamma
    \\[-0.1em]
    &\quad
    + \int_1^\infty \frac{1}{\gamma^2}
       \frac{\partial}{\partial s}
       E_\gamma(B+tH\Vert B+sX)\bigg\vert_{s=0}\, \d\gamma.
    \label{eq:LC_step2}
\end{aligned}
\end{equation}
In the above equation, we interchanged the integral and the partial derivative by an application
of Lebesgue's dominated convergence theorem (see, e.g., \cite[Thm.~2.24]{folland2013real}). 
We will justify this explicitly for the first integral, as the reasoning for the second is entirely analogous. Consider
the difference quotients
\be
f_s(\gamma)
\coloneqq \frac{E_\gamma(B+sX\Vert B+tH)-E_\gamma(B\Vert B+tH)}{s}.
\ee
For each fixed $\gamma$ outside a finite set (of Lebesgue measure zero), the
limit $\lim_{s\to 0} f_s(\gamma)$ exists and coincides with the expression in
\eqref{eq:dE1}. Moreover,
\be
|f_s(\gamma)|
\leq \|X\|_1
\ee
for all $s$ and $\gamma$. Finally, because $B>0$,
Frenkel's integral representation \eqref{eq:Frenkel} implies that $E_\gamma(B+sX\Vert B+tH)$
vanishes for $\gamma$ outside a compact interval $[1,\Gamma]$ independent of
small $s$ and $t$. Hence the family $\{\gamma\mapsto f_s(\gamma)/\gamma\}_s$ is
dominated by the integrable function $\gamma\mapsto \|X\|_1/\gamma$ on $[1,\Gamma]$,
and dominated convergence yields the desired interchange of limit and integral.

Substituting~\eqref{eq:dE1} and~\eqref{eq:dE2} at $s=0$ into~\eqref{eq:LC_step2}, and then differentiating
with respect to $t$ at $t=0$, we obtain
\begin{align}
    \Tr\!\left[ X \cdot \mathrm{D}\log[B](H) \right]
    &= -\frac{\partial}{\partial t}
       \int_1^\infty 
       \Big(
           \Tr\!\left[X\{(1-\gamma)B > \gamma t H\}\right]
           - \Tr\!\left[X\{(\gamma-1)B < tH\}\right]
       \Big)\frac{\d\gamma}{\gamma}\bigg\vert_{t=0}
\end{align}
For the first term and 
$t\in \big(0, t_0/2\big]$,
\begin{equation}
\begin{aligned}
\int_{1}^{\infty}
\Tr\!\left[X \{(1-\gamma)B>\gamma tH\}\right]\frac{\d\gamma}{\gamma}
&=
\int_{1}^{\infty}
\Tr\!\left[X \Bigl\{ \frac{1-\gamma}{\gamma t} B > H \Bigr\}\right]\frac{\d\gamma}{\gamma}\\
&=t\int_{-1/t}^{0}
\Tr[X\{uB>H\}]\frac{\d u}{1+tu},
\end{aligned}
\end{equation}
with 
$u=\frac{1-\gamma}{\gamma t}$.
The projection $\{uB>H\}$ is zero for all $u<-1/t_0$ by our choice of $t_0$. Therefore, the lower limit $-1/t$ can be 
replaced with $-1/t_0$. Since the above expression vanishes at $t=0$, we have
\begin{equation}\begin{aligned}
-\frac{\partial}{\partial t} \int_{1}^{\infty}
\Tr\!\left[X \{(1-\gamma)B>\gamma tH\}\right]\frac{\d\gamma}{\gamma}\bigg\vert_{t=0} &= \lim_{t\to 0} \int_{-1/t_0}^0 \Tr[X\{uB>H\}]\frac{\d u}{1+tu} \\
&= \int_{-1/t_0}^0 \Tr[X\{uB>H\}] \left( \lim_{t\to 0} \frac{1}{1+tu} \right) \d u \\
&= \int_{-1/t_0}^0 \Tr[X\{uB>H\}]\, \d u \\
&= \int_{-\infty}^0 \Tr[X\{uB>H\}]\, \d u,
\end{aligned} \end{equation}
where, on the second line, we took the limit inside the integral using once again Lebesgue's Dominated Convergence theorem, this time with dominating function $\left| \frac{1}{1+tu}\right| \leq \frac{1}{1 - t/t_0} \leq 2$.
%
Similarly,
\begin{equation}
\begin{aligned}
\int_{1}^{\infty}
\Tr\left[X\{(\gamma-1)B<tH\}\right]\frac{\d\gamma}{\gamma}
&=
\int_{1}^{\infty}
\Tr\!\left[
X \Bigl\{ \frac{\gamma-1}{t}B<H \Bigr\}
\right]\frac{\d\gamma}{\gamma}\\
&=t\int_{0}^{1/t_0}
\Tr\left[X\{uB<H\}\right]\frac{\d u}{1+tu},
\end{aligned}
\end{equation}
so that
\begin{align}
\frac{\partial}{\partial t}
\int_{1}^{\infty}
\Tr[X\{(\gamma-1)B<tH\}]\frac{\d\gamma}{\gamma}
\bigg\vert_{t=0}
=
\int_{0}^{\infty} \,
\Tr[X\{uB<H\}]\,\d u.
\end{align}
In the end, we obtain
\begin{align}
    \Tr[X\cdot\mathrm D\log[B](H)]
=
\int_{0}^{\infty}\Tr[X\{H>\gamma B\}]\,\d\gamma
-
\int_{-\infty}^{0}\Tr[X\{H < \gamma B\}]\, \d\gamma.\label{eq:LC_step3}
\end{align}
Since~\eqref{eq:LC_step3} holds for every Hermitian $X$, we conclude that
\begin{align}
    \mathrm{D}\log[B](H)
    &= \int_0^\infty \proj{H>\gamma B}\,\d\gamma
      - \int_{-\infty}^0 \proj{H< \gamma B}\,\d\gamma\\
      &= \int_0^\infty \proj{H>\gamma B}\,\d\gamma
      - \int_{-\infty}^0 \proj{H\le \gamma B}\,\d\gamma.
\end{align}
The proof is complete.
\end{proof}

\begin{lemm}[\hspace{1sp}{\cite[Lemma 2.2]{LHC25_layer_cake}}] \label{lemm:differentiability}
Let $K$ and $L$ be Hermitian matrices.
Then,
\begin{align}
    \frac{\d}{\d t} \Tr\left[ (K-t L)_+ \right]
    = - \Tr\left[ L \left\{ K > t L \right\}\right]
    = - \Tr\left[ L \left\{ K \geq t L \right\}\right],
\end{align}
except for $t$ such that $K-tL$ is singular.
\end{lemm}

\begin{lemm} \label{lemm:Lipschitz}
    The function $Y \mapsto \Tr[Y_+]$ is $1$-Lipschitz continuous on Hermitian operators with respect to the trace norm.
\end{lemm}

\begin{proof}
Let $X$ and $Y$ be Hermitian operators.
Via the variational formula, we have
\begin{align}
\begin{split}
\left| \Tr[X_+] - \Tr[Y_+] \right|
&= \left| \max_{0\leq \Lambda \leq \mathds{1}} \Tr[\Lambda X] - \max_{0\leq \Lambda \leq \mathds{1}} \Tr[\Lambda X] \right|
\\
&\leq \max_{0\leq \Lambda \leq \mathds{1}} \left| \Tr\left[ \Lambda (X-Y) \right] \right|
\\
&= \max\left\{ \max_{0\leq \Lambda \leq \mathds{1}} \Tr\left[ \Lambda (X-Y) \right],  
\max_{0\leq \Lambda \leq \mathds{1}} \Tr\left[ \Lambda (Y-X) \right]
\right\}
\\
&= \max\left\{ \Tr\left[(X-Y)_+\right],  \Tr\left[(X-Y)_-\right] \right\}
\\
&\leq \left\| X - Y \right\|_1. \qedhere
\end{split}
\end{align}
\end{proof}



\section*{Acknowledgments}
LL is grateful to Mario Berta and Bartosz Regula for an early discussion that took place at the ``Spiaggia del Rogiolo'' in June 2025.
The discussion of this note between HC, LL, and PL took place during the workshop ``Mathematics in Quantum Information'' at RWTH Aachen University.
We sincerely thank Mario Berta for the kind hospitality. 
HC is supported by grant No.~NSTC 114-2628-E-002 -006, NSTC 114-2119-M-001-002, NSTC 114-2124-M-002-003, NTU-114V2016-1, NTU-114L895005, and NTU-114L900702. GG acknowledges financial support from the Israel Science Foundation under Grant No. 1192/24. LL acknowledges financial support from the European Union under the ERC StG ETQO (grant agreement
no.~101165230).



{\larger
\bibliographystyle{myIEEEtran}
\bibliography{reference.bib, operator.bib, operator2.bib}
}

\end{document}